\documentclass[%
 aip,
 jmp,
 amsmath,amssymb,
 reprint,
]{revtex4-1}

\usepackage{graphicx}
\usepackage{dcolumn}
\usepackage{bm}
 
\usepackage[utf8]{inputenc}
\usepackage[T1]{fontenc}
\usepackage{mathptmx}
\usepackage{etoolbox}
\usepackage{dsfont} 
\usepackage{amsthm}
\usepackage[english]{babel}

\usepackage{mathtools} 

\usepackage{slashed}

\usepackage{tikz-cd}
\usepackage{dsfont}
\usepackage{verbatim}

\RequirePackage{tikz-cd}
\RequirePackage{amssymb}
\usetikzlibrary{calc}
\usetikzlibrary{decorations.pathmorphing}

\tikzset{curve/.style={settings={#1},to path={(\tikztostart)
    .. controls ($(\tikztostart)!\pv{pos}!(\tikztotarget)!\pv{height}!270:(\tikztotarget)$)
    and ($(\tikztostart)!1-\pv{pos}!(\tikztotarget)!\pv{height}!270:(\tikztotarget)$)
    .. (\tikztotarget)\tikztonodes}},
    settings/.code={\tikzset{quiver/.cd,#1}
        \def\pv##1{\pgfkeysvalueof{/tikz/quiver/##1}}},
    quiver/.cd,pos/.initial=0.35,height/.initial=0}

\tikzset{tail reversed/.code={\pgfsetarrowsstart{tikzcd to}}}
\tikzset{2tail/.code={\pgfsetarrowsstart{Implies[reversed]}}}
\tikzset{2tail reversed/.code={\pgfsetarrowsstart{Implies}}}
\tikzset{no body/.style={/tikz/dash pattern=on 0 off 1mm}}

 \makeatletter 
 \def\@email#1#2{%
  \endgroup
  \patchcmd{\titleblock@produce}
   {\frontmatter@RRAPformat}
   {\frontmatter@RRAPformat{\produce@RRAP{*#1\href{mailto:#2}{#2}}}\frontmatter@RRAPformat}
   {}{}
 }%
 \makeatother

\newtheorem{theorem}{Theorem}
\newtheorem{definition}{Definition}
\newtheorem*{remark}{Remark} 
\newtheorem*{remarks}{Remarks}

\newcommand{\com}{\ensuremath{\mathds{C}}}
\newcommand{\real}{\ensuremath{\mathds{R}}}

\newcommand{\glc}{ \ensuremath{\mathrm{GL}(d_\gamma,\mathds{C} )}}
\newcommand{\slc}{ \ensuremath{\mathrm{SL}(d_\gamma,\mathds{C} )}}
\newcommand{\glr}{ \ensuremath{\mathrm{GL}(D,\mathds{R} )}}
\newcommand{\uone}{ \ensuremath{\mathrm{U}(1)}}

\newcommand{\so}{ \ensuremath{\mathrm{SO}}}
\newcommand{\spin}{ \ensuremath{\mathrm{Spin}}}
\newcommand{\pin}{ \ensuremath{\mathrm{Pin}}}
\newcommand{\cl}{ \ensuremath{\mathfrak{Cl}}}
\DeclareMathOperator{\Dim}{dim}

\begin{document}


\title{Local spin base invariance from a global differential-geometrical point of view }
\author{Claudio Emmrich}
\email[]{claudio.emmrich@uni-jena.de}
\affiliation{Theoretisch-Physikalisches Institut, Friedrich-Schiller-Universit\"at Jena, Max-Wien-Platz 1, 07743 Jena, Germany}

\date{\today}

\begin{abstract}
	This article  gives a geometric interpretation of the spin base formulation with  local spin base invariance of spinors on a curved space-time   and in particular of a central element, the global Dirac structure, in terms of principal and vector bundles and their endomorphisms. 
    It is  shown that this is intimately related to $\spin$ and $\spin^\com$ structures in the sense that the existence of one of those implies the existence of a Dirac structure and allows an extension to local spin base invariance. Vice versa, as a central result, the existence of a Dirac structure implies the existence of a $\spin^\com$ structure.  Nevertheless, the spin base invariant setting may be considered more general, allowing more physical degrees of freedom. Furthermore,  arguments are given that the Dirac structure is a more natural choice as a variable for (quantum) gravity than tetrads/vielbeins.
\end{abstract}  
	

\maketitle 

\section{Introduction}
	 The starting point of local spin base invariance is  the elementary observation concerning spinors that  the relation defining the $\gamma$ matrices   and the corresponding Clifford algebra: $\{ \gamma_\mu, \gamma_\nu\} = -2 g_{\mu,\nu}  \openone$ is invariant under similarity transformations.  Local spin base invariance is the extension of this global symmetry to a local symmetry for spinors on arbitrary (curved) space-times: Since no derivatives are involved, the choice of similarity transformations may be  extended to transformations varying over space-time in a rather straightforward way. The natural variables for this formulation are space-time-dependent Dirac matrices subject to the Clifford-algebra constraint. No vielbeins are needed as opposed to the standard approaches to spinors  on curved space-times   used in most physics literature.

	 This idea has a rather long history, going back to Schrödinger \cite{schrodingerDiracSchesElektron1932} and Bargmann \cite{bargmannBemerkungenZurAllgemeinrelativistischen1932} in 1932.  The formalism has been applied later to quantization of fermions in a curved background metric in Refs.~\onlinecite{brillInteractionNeutrinosGravitational1957a,unruhSecondQuantizationKerr1974,casalsKermionsQuantizationFermions2013}. Though most literature on spinors on curved space-times focuses on the vielbein approach, the theory has been further developed in the last decades, see e.g. Refs.~\onlinecite{finsterLocalSymmetryRelativistic1998,weldonFermionsVierbeinsCurved2000}. An overall review of the formalism, including further literature, and also including spin torsion may be  found  in    Ref.~\onlinecite{giesFermionsGravityLocal2014}. In particular, arguments why the Dirac matrices might be considered more adequate variables for a quantization of gravity than vielbeins, both from a conceptual and a pragmatic (simplicity) point of view  may be found in
	 Refs~\onlinecite{giesFermionsGravityLocal2014,giesGlobalSurplusesSpinbase2015,lippoldtFermionischeSystemeAuf2012} and an extension to arbitrary dimensions in     Ref.~\onlinecite{lippoldtSpinbaseInvarianceFermions2015}. One surprising feature of the whole approach is that the local Lorentz or diffeomorphism transformations acting on the space-time index of the $\gamma_\mu$ and the spin base transformations acting on the ``matrix part'' by conjugation decouple. Thus, in a proper sense, spinors behave as scalars under Lorentz transformations. 
	 
	 Despite the benefits of this approach, some conceptual questions have remained, partially because the approach is formulated in a local, coordinate and spin base dependent way (though a well-defined transformation behaviour under coordinate and base changes is defined):
	 \begin{itemize}
	   \item Do the Dirac matrices correspond to a global geometric object?
	   \item Are there any global obstructions towards this formalism? In particular, how does this formalism relate to the global geometrical approaches using $\spin$ or $\spin^\com$ structures (see e.g.    Refs.~\onlinecite{lawsonSpinGeometryPMS382016a,figueroa-ofarrillSpinGeometry2017}).
	   \item What does the statement ``spinors transform like scalars under Lorentz transformation'' really mean? How does it fit to the central role played by the Lorentz group  for spinors on flat space times in  quantum mechanics  and quantum field theory?  
	   \item From a more pragmatic point of view:  In  Ref.~\onlinecite{giesGlobalSurplusesSpinbase2015} Gies and Lippoldt showed that there is a global realization of the Clifford algebra on a 2-sphere, which is not possible within the vielbein formalism. Does this generalize? Is there a geometric origin for this simplification?
	 \end{itemize}
	 
	 In this article, we will answer those questions by giving a global geometric description of the spin base formalism.  It will turn out that 
	 \begin{itemize}
	     \item The Dirac matrices actually correspond to a geometric object, namely a global section in a suitable bundle, denominated ``Dirac structure'' in the following.  It exists  in all cases where spinors are  defined ($\spin$ structure, $\spin^\com$ structure, and  in the case of the spin base formalism as a basic requirement).  This is different from vielbeins, which are by definition local sections in the bundle of (pseudo)-orthonormal frames, and which exist globally only iff this bundle is trivial, i.e. iff the space-time is parallelizable (see also remark below).  
	     \item The spin base formalism exists if and only if a $\spin^\com$ structure exists: If the latter is defined, one may derive a corresponding Dirac structure, conversely, if a Dirac structure is defined, one may construct from it a $\spin^\com$ structure essentially by reduction of structure groups using the Dirac structure (in a physical language by partial gauge fixing). 
	     \item Thus, the Lorentz transformations turn out to be intertwined with the local spin base transformations by the mere existence of a Dirac structure in a rather subtle way.
	     \item The additional gauge freedom can be used to simplify calculations as shown for the $S^2$-example by Gies and Lippoldt. We generalize this setting to arbitrary spheres $S^n$ and show that the enhanced spin base symmetry can be used to trivialize the spinor bundle over this larger group. It turns out that this is related to the fact that the Whitney sum of the tangent bundle and the normal bundle is trivial and that Clifford algebras in dimension $n$ and $n+1$ are closely connected. 
	 \end{itemize}
	 
	 Despite the fact that local spin base invariance does not avoid the known obstructions to the existence of  $\spin^\com$ structures\cite{lawsonSpinGeometryPMS382016a}, the approach is nevertheless not redundant and worth studying for at least two reasons:
	 \begin{itemize}
	     \item The global existence of a Dirac structure (whenever spinors may be defined) is another indication that it may be a better candidate for a variable in the quantization of general relativity (GR) than vielbeins: The latter do normally (unless space-time is parallelizable) only exist locally, the only ``really existing'' global objects behind them are the bundle of orthonormal frames and the metric itself, which always exist on any space-time, completely irrespective of whether spinors may be defined on this space-time or not. So the frequent claim in physics literature that vielbeins are needed because fermions exist in nature does turn out to be rather  an argument in favour of Dirac structures, and not in favour of vielbeins.
	     \item Despite the fact that the involved  bundles can be reduced to a $\spin^\com$ structure, the more general setting of Dirac structures with local spin base invariance allows for generalizations of general relativity and additional degrees of freedom, e.g. by inclusion of spin torsion as in Ref.~\onlinecite{giesFermionsGravityLocal2014} and may have an impact for example on theory space for a  functional renormalization group approach.
	 \end{itemize}

	\begin{remarks} ~ 
	\begin{itemize}
		\item   Geroch  has shown in Ref.~\onlinecite{gerochSpinorStructureSpacetimes1968} that  non-compact Lorentzian manifolds in 4 dimensions are always parallelizable. Hence, the whole discussion about global aspects, existence of global vielbeins,...  is   mandatory only in  more general cases: This includes non-Lorentzian manifolds, in particular the seemingly simpler case of Riemannian  manifolds, compact Lorentzian manifolds or Lorentzian manifolds of dimensions higher than 4, widely studied in physics literature (e.g. for Kaluza-Klein like theories or string theory). However, since this parallelization is not canonical, and the proof shows that this is loosely speaking a coincidence for Lorentzian manifolds in dimension 4 or less, an approach avoiding global vielbeins and using this parallelizability may be conceptually preferable even in the four dimensional Lorentzian case. 
		(The proof is   based  on the coincidence of restrictions from obstructions towards existence of Lorentzian metric on any  manifolds, and the fact that in 4 dimensions, concerning the homotopy groups of  $\mathrm{SL}(2,\com)$,     ``The third homotopy group fails to vanish, but at this point we are sufficiently close to the dimension of the manifold that the obstruction to extending a cross section can be made to vanish''   \cite{gerochSpinorStructureSpacetimes1968}).  
		\item It is well known that a Clifford bundle may be defined globally without any additional structure and any obstructions on any Riemannian or Lorentzian manifold (see appendix \ref{sec:app_prop_cliff}). Hence, the central structure here is the Dirac structure which connects the (co)tangent space of the manifold with endomorphisms in a complex vector bundle in a way compatible with the Clifford algebra (thus defining a representation of the Clifford bundle on a spinor bundle, see appendix \ref{sec:dirac_cliff_bundle}).
		\item A related approach with a  much further going generalization of spin bundles is considered in  Ref. \onlinecite{finsterSpinorsSingularSpaces2019}:
		In the context of so called causal fermion systems even non smooth generalizations of spin bundles are allowed. Here, even in the smooth setting, where fibre bundles are regained, the fibres have a (pseudo) scalar product which a priori is independent of any (Riemannian or Lorentzian) metric on the base manifold. However, if one imposes additional constraints to achieve this connection, a classic spin structure with it's topological restrictions (vanishing second Stiefel-Whitney class) is restored, with additional restrictions if the existence of a Clifford section is required. However, those constraints  lead to a $\spin$ structure  and  not to the more general  $\spin^\com$ structures naturally arising in our approach.
		\item In the following, we will restrict ourselves to the case of irreducible representations of the Clifford algebra, so essentially to a single fermion. This excludes Kähler fermions (Ref. \onlinecite{kahlerInnereDifferentialkaklul,Graf1978DifferentialFA,banksGeometricFermions1982,bennFermionsSpinors1983}), which are known to exist on any (pseudo-)Riemannian manifold, but always come as multiplets of e.g. four Dirac-fermions in four dimensions.
	\end{itemize}
	\end{remarks}
		
	The overall setting of this article is as follows: In the next section we give a very short introduction to the spin base approach, focussing only on those aspects needed for a global geometric view and to address the topics above. In particular, the central object, the  Dirac structure, is given a global geometric definition. 
	
	In section \ref{sec:dirac_from_spin}, the standard approach to spinors using $\spin$ or $\spin^\com$ is briefly sketched as well in order to lay the foundation to relate them to the spin base approach in the following sections.  We show that the standard $\spin$/$\spin^\com$ formalisms have a naturally defined Dirac structure and allow for the extension of the structure groups to yield the local spin base invariant setting.
	
	In section \ref{sec:reduction}, which is the central part of this paper,  we show that conversely, the existence of a global Dirac structure over a pseudo-Riemannian manifold ensures the existence of a $\spin^\com$ structure. The whole construction is natural in the sense that if one starts from a Dirac structure and constructs the corresponding  $\spin^\com$ structure, then  the   Dirac structure  corresponding to this $\spin^\com$ structure  is the original Dirac structure.
	
	In section \ref{sec:metric_conn},   we show how the existence of a metric on the spinor bundle and a $\gamma$ compatible connection may be easily derived, under some slight restriction,  from the results of the previous section and some elementary facts about Clifford algebras over a fixed vector space.

	In section \ref{sec:sphere}, we generalize and make more precise the results of Gies and Lippoldt on the 2-sphere to arbitrary spheres by using global geometric arguments without any tedious coordinate computations, thus supporting the view that the spin base approach does not only have conceptual, but also pragmatic calculational advantages.
	
	Finally in  \ref{sec:conclusion}, we give some conclusions, and in the appendices we collect some useful facts about Clifford algebras and principal bundles as a reference for readers without firm background in those topics.

\section{The idea of local spin base invariance in geometric terms}
	
\subsection{The idea of spin base invariance in local coordinates}
	In most approaches to spinors on curved space time, the central focus is on so called tetrads or vielbeins, a (smooth) choice of local Lorentz (= pseudo-orthonormal) frames over each space-time point. The spin group (the double cover of the Lorentz/pseudo-orthogonal group $\so(r,s)$)  acts both on the vielbeins  and on the spinors, which locally can be identified with  complex vector valued functions (For a geometric formulation using principal and vector bundles see section \ref{sec:spin_structs}). Both are intertwined by the property of Dirac matrices  
	\begin{equation} \label{eq_flatdirac_trafo}
		\rho(g) \gamma_{a} \rho(g)^{-1}=(\Lambda(g)^{-1}) ^{b} _{~~ a}\gamma_b ,
	\end{equation}
	where $	\rho(g) $ is a spin representation of a Lorentz transformation (more precisely, of an element $g$ of the spin group  covering this Lorentz transformation)  and $\Lambda (g) $ the vector representation of the same group element.
	The right-hand side transforms as a co-vector with respect to the space-time index, since we are using $\gamma$ with a lower index. (By raising the indices with the metric, and using the definition of Lorentz transformations, an analogous relation holds for $\gamma^a$ with upper index, transforming as a vector.) 
	
	The defining property for the Dirac matrices $\gamma_{a}$ is the Clifford relation: 
	\begin{equation} \label{eq:flatdirac_def}
	 \{ \gamma_{a}, \gamma_{b}\} = - 2 \eta_{ab} \openone,
	\end{equation} 
	where $\eta$ is diagonal with $r$ elements $1$ and $s$ elements $-1$ ($s=1$ in the Lorentz case), and $\openone$ denotes the identity matrix. 
	As is well known, those matrices are unique up to similarity transformations and (in the case of odd dimensions) a sign flip and are of size $d_\gamma \times d_\gamma$ with $d_\gamma=2^{\lfloor D/2\rfloor}$. (Here $\lfloor \ldots\rfloor$ denotes the floor function, see appendix \ref{sec:app_prop_cliff}).
	
	\begin{remark}We use the sign convention of Refs.~\onlinecite{lawsonSpinGeometryPMS382016a,figueroa-ofarrillSpinGeometry2017} on the right hand side used in most of the mathematical literature. The plus sign chosen in most of the physical literature corresponds to exchanging $r$ and $s$. This has no impact on $\spin(r,s)$, since $\spin(r,s)$ and $\spin(s,r)$ are isomorphic, however it has to be taken into account when considering the real Clifford algebra $\cl(r,s)$ where $\cl(r,s) \not\cong \cl(s,r)$ . 
	\end{remark}
	To define spinors as physical fields, one needs a Dirac operator and  hence a connection, which can be locally lifted from the Levi-Civita-connection via the vielbeins,   and $\gamma$ matrices fulfilling Eq.~(\ref{eq:flatdirac_def}). The latter can be simply chosen as a fixed set of constant matrices.
	
	The spin base formalism is an extension and modification of this formalism, where no vielbeins are needed in the local coordinate formulation. It is based on the following observations:
	\begin{itemize}
	    \item Whenever a set of Dirac matrices $\gamma_{a}(x)$ fulfills Eq.~(\ref{eq:flatdirac_def}) for all $x \in M$, then $ S(x) \gamma_{a}(x) S(x)^{-1}$ fulfills this relation as well for any $\glc$-valued function $S(x)$.
	    \item Hence, we may decouple the matrix transformation part from the space-time index trans\-formation part, allowing for arbitrary local frames in the tangent space as base, not only orthonormal ones. In particular, one may choose a set of   holonomic base vectors $\partial_\mu = \frac{\partial}{\partial x^\mu}$ if one replaces $\eta_{ab}$ by $g_{\mu\nu}$ in Eq.~(\ref{eq:flatdirac_def}).
	    \item Since $g_{\mu\nu}$ is completely determined by $\gamma_\mu(x)$, the latter may serve as degrees of freedom  for quantum gravity as long as the Clifford constraint is imposed.
	\end{itemize}
	
	Hence, the formalism is formulated in terms of space-time dependent $\gamma$ matrices  fulfilling the Clifford 	condition
	\begin{equation} \label{eq_cliff}
		\{\gamma_\mu(x), \gamma_\nu (x)\} =-  2 g_{\mu \nu}(x)   \openone
	\end{equation}
	and transforming as a co-vector in the space-time index with respect to arbitrary local coordinate transformations on $M$  and under the conjugation by  $S(x)$ for an arbitrary $\glc$-valued function $S(x)$. For consistency, the spinors themselves must transform under $S(x)$ as well, whereas the local coordinate transformations have no impact on the spinor (if considered as passive coordinate transformations). Hence,  the spinors are treated as  scalars under coordinate transformation in this sense.
	
	The whole formalism may be extended by requiring the existence of a compatible (indefinite, but non-degenerate) metric on the spinors and of a connection compatible with $\gamma_\mu$. As shown in detail in    Refs.~\onlinecite{weldonFermionsVierbeinsCurved2000,giesFermionsGravityLocal2014,lippoldtFermionischeSystemeAuf2012,lippoldtSpinbaseInvarianceFermions2015}, this allows the formulation of spinors and of general relativity  in terms of the new coordinates without using vielbeins. It allows to study extensions of GR, e.g. by allowing for ``spin torsion''. Furthermore, arguments based on path integral measures and renormalization group considerations are given there to support the view that those coordinates might be better  coordinates than vielbeins.

\subsection{Spin base invariance and geometry} \label{sec:spin_inv_geom}
	As already indicated in the introduction, the formulations of the spin base formalism in the literature is based on a local coordinate consideration (including transformation under coordinate changes). Thus, questions on the global existence of the formalism and its objects for arbitrary (pseudo)-Riemannian manifold, potentially depending on global (topological)  aspects cannot be easily addressed in this formalism. Furthermore, some of the proofs require rather involved computations with $\gamma$ matrices.
	
	For this reason, we first translate the formalism into the language of principal bundles and associated vector bundles    \cite{kobayashiFoundationsDifferentialGeometry2009,eguchiGravitationGaugeTheories1980, nakaharaGeometryTopologyPhysics2018, steenrodTopologyFibreBundles1999,baez1994gauge} (see also appendix \ref{sec:conn_principal}). This allows to answer those questions and to show that the formalism is indeed globally well-defined (under appropriate conditions). Furthermore, this will give additional hints that the Dirac matrices may be indeed preferable coordinates as compared to vielbeins.

	The central component of the spin base invariance approach are the  $\gamma_\mu(x)$ forming a Clifford algebra,  
	which transforms in the space-time-index $\mu$  like a co-vector in $T^*M$ and as a matrix  by  conjugation with a separate $\glc$. 
	
	We now translate this into a geometric language: Since the space-time and the spin indices transform separately, this means $\gamma_\mu$ must lie in the tensor product of two different vector bundles. The first, corresponding to $\mu$ obviously is $T^*M$ (or $TM$ if we prefer to consider $\gamma^\mu$ instead of $\gamma_\mu$, which can be easily translated into one another using $g$). The second bundle must be related to a complex vector bundle of dimension $d_\gamma$, but since $\gamma^\mu$  does not transform with the defining representation of $\glc$, but by conjugation,   this means it is the bundle of endomorphism of a complex vector bundle $E$ of dimension $d_\gamma$.

	Hence, we are lead to the following definition: 
    \begin{definition}[Dirac Structure]
        Let $(M,g)$ be a pseudo-Riemannian manifold of dimensions $D$ and $E$ a complex vector bundle over $M$ of dimension $d_\gamma=2^{\lfloor D/2\rfloor}$ . Let $\{,\}$ denote the fibre-wise anticommutator of endomorphisms of $E$ and $ \openone \in \Gamma( \mathrm{End}(E))$ denote the identity isomorphism on each fibre. A \textbf{Dirac structure} is a global section of ${T^*M \otimes \mathrm{End}(E)}$ such that
        \begin{equation}
            \{\gamma(X),\gamma(Y)\} = -  2 g(X,Y)  \openone    \label{eq:dirac_str}
        \end{equation}
        for any two vector fields $X,Y \in \Gamma(T M)$.
    \end{definition}

    \begin{remarks} ~
        \begin{itemize}
         \item  As stated before, we restrict ourselves to the case of irreducible representations of the Clifford algebra.
         \item 
            By a polarization argument, it is sufficient to request $ \gamma(X) \circ \gamma(X) = - g(X,X)  \openone$ for all vector fields  $X \in \Gamma(T M)$, which is closer to the way      Clifford algebras are mostly defined in the mathematical literature.
        \end{itemize}
    \end{remarks}

	With this definition, the central requirement of the spin base formalism is the existence of a global object, namely the Dirac structure, on which we will focus in the following two sections. The two additional structures needed to define a reasonable physical theory, namely the spin metric and a compatible connection, will turn out to exist globally  whenever the Dirac structure is defined, hence we defer those to section \ref{sec:metric_conn}. 
	
	\begin{remark}
	One might argue that the local existence of a Dirac structure might be sufficient. However, this is not the case: It is a central component of the spin base formalism, in particular needed to define a compatible connection. As we will show in \ref{sec:metric_conn}, the compatibility condition for the connection, despite looking very similar to a seemingly corresponding equation in the vielbein formalism, the so called ``vielbein postulate'',  has a completely different meaning: whereas the latter is a (still computational useful) triviality, namely the statement
	$\nabla \textrm{id}=0$ (see e.g. appendix J of    Ref.~\onlinecite{carrollSpacetimeGeometryIntroduction2019}  or section \ref{sec:metric_conn}), the former is an essential compatibility condition implementing a slightly stricter condition than metricity. Furthermore, the existence of a global Dirac structure is necessary for the existence of  a complex vector bundle carrying a representation of the Clifford bundle. Finally, since it exists whenever a $\spin^\com$ structure is defined, so in particular  whenever a $\spin$ structure is defined, the existence of a Dirac structure is a central component of any approach to spin on curved manifolds. 
	\end{remark}

	Note that we do not assume a priori any connection between the vector bundle $E$ and $TM$, in terms of local coordinates:  transition functions may be defined for $E$ completely independently from the transition functions in $TM$ induced by the respective coordinate mappings, thus implementing full spin base invariance.
	
	However, it turns out that the existence of the Dirac structure $\gamma$ with its  Clifford algebra condition are so strict that it does intertwine the transition functions of $E$ and $TM$ to some extent, so stating that spinors act like scalars under coordinate transformations turns out to be correct only with some mild caveat.

\section{Locally invariant Dirac structure from Spin and Spin$^\com$ ~structures } \label{sec:dirac_from_spin}
\subsection{Spin and Spin$^\com$ ~structures} \label{sec:spin_structs}
    For completeness and in order to fix notation, we briefly recapitulate the standard definitions of Spin and Spin$^\com$ ~structures as needed in the following (see e.g.  Refs.~\onlinecite{lawsonSpinGeometryPMS382016a, figueroa-ofarrillSpinGeometry2017, nakaharaGeometryTopologyPhysics2018}) before showing the existence of a Dirac structure in those standard approaches to spin: 
    
    Let  $(M,g)$ be an orientable pseudo-Riemannian manifold with signature $(r,s)$. All manifolds are assumed to be connected, orientable, smooth, paracompact, Hausdorff. We restrict  the following considerations to the case $D\coloneqq \Dim(M) =   r+s \geq 3$ and keep $r,s$ fixed. To avoid cluttered notation, the indices $r,s$ are not explicitly added for most objects, so for example $\so(M)$ denotes the bundle of pseudo-orthogonal frames with respect to $g$. 
    
    We define $\spin(r,s)$ as the Lie group forming a double covering $\Pi: \spin(r,s) \rightarrow \so(r,s)$ of $\so(r,s)$ with $\Pi(g h ) = \Pi(g) \Pi(h)$.  It may be obtained as the even product of elements of norm $1$ of the underlying vector space of the Clifford algebra (see appendix \ref{sec:app_prop_cliff}). For the Riemannian and the Lorentzian case, the component of the identity $\spin_0(r,s)$ is known to be simply connected and thus to form the universal cover of $\so_0(r,s)$, 
    \begin{remark}
         This is equivalent to a definition  by the exact sequence of group homomorphisms: 
         \begin{equation} \label{eq:cov_spin} 1 \rightarrow \mathds{Z}_2  \rightarrow \spin(r,s)  \xrightarrow{\Pi} \so(r,s) \rightarrow 1  \quad .\end{equation}
    \end{remark}

    A spin structure on $(\mathrm{M}, g)$ is a principal $\spin(r, s)$  bundle $\spin(\mathrm{M}) \xrightarrow{\tilde{\pi}} \mathrm{M}$ together with a bundle morphism $\Phi: \spin(M) \rightarrow {\so(M)}$
    which restricts fibrewise over each point in $M$ to the covering homomorphism  $\Pi: \operatorname{Spin}(r,s) \rightarrow \mathrm{SO}(r,s)$, or more explicitly: 
    \noindent an equivariant bundle morphism $\Phi:$ $\Phi(p g) = \Phi(p) \Pi(g) ~ \forall p\in  \spin(\mathrm{M}), g \in \spin(s, t)$ such that $\pi (\Phi(p)) = \tilde{\pi}(p) $ for all $ p\in  \spin(\mathrm{M})$, i.e., a point $p$  in a fibre of \spin(M)  over a point $x\in M$ is mapped to a point in a fibre of $\so(M)$ over the same point $x$:
    \begin{equation}\begin{tikzcd}[sep=tiny]
    	{\spin(M)} & {} & {\so(M)} \\
    	& M
    	\arrow["\Phi", from=1-1, to=1-3]
    	\arrow["\pi", from=1-3, to=2-2]
    	\arrow["{\tilde{\pi}}"'', from=1-1, to=2-2]
    \end{tikzcd}\end{equation}

    The principal bundle $\tilde{\pi}: \spin(\mathrm{M}) \rightarrow M$ is also called the bundle of spin frames over $M$.
    From a physical point of view, this rather technical global formulation essentially means that the ambiguity in sign, resulting from the fact that spinors change sign under a rotation by $ 2 \pi$ may be resolved for the transition functions of local charts in a consistent way. This need to choose consistently a branch in the double covering of $\so(r,s)$ on an overlap of three coordinate neighbourhoods leads to the known obstruction that the second Stiefel-Whitney class of $M$ has to vanish in order for a spin structure to exist.
    
    This obstruction may be somewhat weakened by considering a more general structure, the $\spin^\com$ structure. For its definition, we first need the notion of the $\spin^\com$   group:

    It is the quotient $\spin^{\com}(r,s)=\left(\spin(r,s) \times \uone\right) /  { \mathds{Z}_2}$  with $\mathds{Z}_2$ acting  as  $(g,h) \mapsto (-g,-h)$.  Equivalently , it is defined by the exact sequence:  
    \begin{equation} \label{eq:def_spinc}
        1 \rightarrow \mathrm{Z}_{2} \rightarrow \operatorname{Spin}^{\mathrm{C}}(r,s) \xrightarrow{\rho} \mathrm{SO}(r,s) \times \mathrm{U}(1) \rightarrow 1 \quad.\end{equation}
    The mapping $\rho$ may be explicitly defined  as $\rho( [(g,u)]) \coloneqq (\Pi(g), u^2)$, where $[(g,u)]$ denotes the $\mathds{Z}_2$-equivalence class of $(g,u) \in \left(\spin(r,s) \times \uone\right)$. This is obviously well-defined since the sign ambiguity in choosing a representative $(g,h) \in [(g,h)]$ drops out on the right hand side. 
    
    A $\spin^\com$ structure on $SO(M)$ consists of a principal $\mathrm{U}(1)$   bundle $\mathrm{U}_{1}(M)$ over $M$ and  a principal $\spin^\com(r,s)$  bundle $\spin^\com(M)$ with a $\spin^\com(r,s)$-equivariant bundle map
    \begin{equation} \label{eq;spin_com_cov}
    \begin{tikzcd}[sep=tiny]
    	{\spin^\com(M)} && {\so(M)\times U_1(M)} \\
    	& M
    	\arrow[from=1-1, to=2-2]
    	\arrow[from=1-3, to=2-2]
    	\arrow["{\Phi^\com}", from=1-1, to=1-3]
    \end{tikzcd}
    \end{equation}
    Intuitively, the additional $\mathrm{U}(1)$  bundle may be considered as corresponding to a $\mathrm{U}(1)$ charge which leads to an additional phase which in some cases can compensate the obstruction of choosing ``the right branch'' in the double cover of $\so(r,s)$. However, even for $\spin^\com$ structures there is a (weaker)  obstruction preventing it from existing on all manifolds ( the second Stiefel-Whitney-class  $w_{2}(\so(M))$ must be the mod 2 reduction of an integral class, see   Ref.~\onlinecite{lawsonSpinGeometryPMS382016a}).

    For a $\spin$ structure we can define the associated vector bundle, the spinor bundle,  of dimension
    $d_\gamma \coloneqq \lfloor\frac{D}{2}\rfloor$:

    \begin{equation} \label{eq:assoc_vb_sp}
        E \coloneqq \spin(M) \times_{\spin(r,s)}  \com^{d_\gamma}  =  \left( \spin(M) \times  \com^{d_\gamma}\right) / \thicksim
    \end{equation}	
	 where the equivalence relation $\thicksim$ is defined by   $\forall g \in   {\spin(r,s)}: ~(p,v)\sim (p g,g^{-1}v) $, with a completely analogous definition for the $\spin^\com$ case.

	The main motivation for considering  $\spin^\com$ instead of $\spin$ is that the obstruction towards existence of such bundles is much weaker. In particular, it is known that any 4 dimensional manifold admits a  $\spin^\com$ structure, which is not true for  $\spin$. (Though in the non-compact Lorentzian case, any 4 dimensional manifold is parallelizable and hence admits a spin structure    \cite{gerochSpinorStructureSpacetimes1968}.)

\subsection{Existence of a global Dirac structure for \spin- and $\spin^\com$ structures} \label{sec:exist_dirac_spin}

	\begin{theorem}
	  A Dirac structure $\gamma$  exists,  whenever a $\spin$ or $\spin^\com$ structure over $(M,g)$ is defined. It is naturally defined and unique up to an overall global $\glc$ transformation and, only in the odd dimensional  case, a potential sign flip. 
    \end{theorem}

    The proof  is a  geometric reformulation in the spirit of ``locally spin base invariant Dirac structures'' of the well known fact that the Dirac matrices in the vielbein language transform as co-vectors with respect to their space-time index  under Lorentz transformations: 
    
    \begin{proof}[Proof:]
	
	For simplicity, we start with the $\spin$ case  and then show  that it almost trivially extends to  $\spin^\com$. So we assume we are given a $\spin(r,s)$ bundle $\spin(M)$ over  $M$.

    Now, it is well known  (see appendix \ref{sec:app_prop_cliff}) that $\com^{\lfloor\frac{D}{2}\rfloor}$ carries an irreducible  $\pin(r,s)$ representation which restricts to a $\spin(r,s)$ representation. (For $D$ even, this  $\pin(r,s)$representation is unique up to isomorphism, but reducible as a $\spin(r,s)$ representation, for $D$ odd, there are two inequivalent representations, both of which are irreducible as $\spin(r,s)$ representations).
  
    Let $E$ be the the associated vector bundle according to Eq.~(\ref{eq:assoc_vb_sp}).

	To define $\gamma$ globally, we first choose a local section of $\spin(M)$, i.e. a smooth  map $\chi_U: U  \rightarrow \spin(M)$ which maps each $x\in U$ to an element of the fibre of $\spin(M)$ over $x$  for an open subset $U \subset M$, define it there and then show that it does not depend on the choice of the local section by considering transition functions.
	
	A local section $\chi_U$ over $U$ of the principal bundle $\spin(M)$ defines a local trivialization $\varphi_U: \tilde{\pi}^{-1}(U) \rightarrow U  \times \spin(r,s)$ via \begin{equation} 
	\varphi_U(\chi_U(x) ) = (x,e) ,  ~~ \varphi_U(\chi_U(x) \cdot g) \coloneqq (x,g) \label{eq:def_triv}\end{equation}
	for any $x\in U, g \in \spin(r,s)$, 
	where $e$ is the identity in $\spin(r,s)$. This is  well  defined, since by definition of a principal bundle, the structure group acts by a free and transitive action from the right on each fibre. 
	
	By projection of principal bundles $\Phi: \spin(M) \rightarrow \so(M)$,  $\chi_U$ defines a local section in $\so(M)$, i.e. a Lorentz frame (vielbein) in $T_xM$ over each point in $U$. In addition, it defines a local trivialization of the associated complex spinor bundle $E$. 
	Define $\gamma |_U$ by a fixed chosen standard representation of the Dirac matrices as constant matrices on the vector space $ \com^{d_\gamma}$ in the induced local trivializations $U \times \com^{d_\gamma}$ of $E$ over $M$ and $U \times \real^{D}$ of $T^*M$ . 
	
	Now, for another open neighbourhood $V$ the  transition function $\Psi_{UV}: U  \cap V \rightarrow \spin(r,s)$ is defined
	by $\varphi_V(p) = \Psi_{UV}(\tilde{\pi}(p)) \cdot \varphi_U(p)$, where the group action ``$\cdot$'' is the left multiplication of the second factor $g$ in $(x,g) \in (U\cap V) \times \spin(r,s) $, leaving the base point $x\in U \cap V$ invariant, or more explicitly: 
	\[ \begin{split}
	\Psi_{UV}&: U  \cap V \rightarrow \spin(r,s)\\ \Psi_{UV}(x) &\coloneqq \pi_2(\varphi_V \circ \varphi_U^{-1} (x,e)),
	\end{split}\]
	where  $ \pi_2: V \times \spin(r,s) \rightarrow \spin(r,s) $ is just the projection on the second factor $\spin(r,s)$. Equivalently, $\Psi(x) $ may be expressed directly via the local sections as: 
	\begin{equation}  \chi_V(x) =\chi_U(x) \cdot \Psi^{-1}_{UV}(x) , \end{equation} as may be easily derived from Eq.~(\ref{eq:def_triv}).
	Note that the transition functions act on the local trivializations from the left, but on the section from the right with an inverse, as is needed for consistency. 
	Hence, the Dirac matrices  $ \gamma_i$ transform as:
	
	\begin{equation} \label{eq_global_gamma}
		\tilde{\gamma}_i(x) =  \left(\Pi(\Psi_{UV}(x))^{-1}\right)^{k}_{~i}  
		~ \Psi_{UV}(x) \cdot \gamma_k \cdot \Psi_{UV}(x)^{-1}  = \; \gamma_i
	\end{equation}
	where    we have used  the property Eqs.  (\ref{eq_flatdirac_trafo}) or (\ref{eq:inv_gamma}) of flat Dirac matrices.
	
	This means that the transformed $\gamma$ has precisely the same form and value (namely the same  constant standard Dirac matrices) in the new coordinates on $V$.  Hence the procedure defines a global well-defined section    $\gamma \in \Gamma( T^*M  \otimes \mathrm{End}(E))$ over all of $M$ independently of the local trivializing section $\varphi$.

	The only choice is the choice of fixed constant Dirac matrices which is known to be unique up to an overall global $\glc$ transformation and a potential global sign flip (the latter occurs only in the odd dimensional case).
	
	For a $\spin^\com$ structure, the same construction works with a few minor modifications: The structure group is no longer $\spin$ but $\spin^\com \coloneqq (\spin \times \uone)/\mathds{Z}_2 $.
	
	Since the $\uone$ factor is in the center of $GL(d_\gamma)$ this factor has no impact and drops out in Eq.~(\ref{eq_global_gamma}) (Care must be taken due to the quotient by $\mathds{Z}_2$.  However, since a factor of $(-1) \in \mathds{Z}_2$ always appears twice and hence cancels as well). 
	\end{proof}

\subsection{Extension of $\spin$/$\spin^\com$ to spin base invariant Dirac structure} \label{sec:ext_spin_dirac}

    In this section we explicitly show the (almost obvious) fact that, given a $\spin^\com$ structure (or as a special case a $\spin$ structure) we may always extend the structure group on the involved bundles such that we end up with the ``spin base invariant Dirac'' structure:

	\begin{theorem}\label{thm:ext} Let $(M,g)$ be an orientable pseudo-Riemannian manifold with a $\spin$ or $\spin^\com$ structure. 
	\begin{itemize}
	   \setlength{\itemsep}{2pt}
		\item The structure group of the complex vector bundle $E$ may be extended from $\spin(r,s) $ or $\spin(r,s)^\com $ to $\glc$ defining a bundle $E^G$ with structure group  $\glc$  (which, as a manifold, is the same as $E$).
		\item $T^*M \otimes \mathrm{End}(E^G)$ transforms under $\glr \times \glc$  as structure group, i.e. space time index and spinor indices transform separately (the endomorphism part obviously transforming by conjugation with  $\glc$. 
		\item $\gamma$ naturally extends to a section in $T^*M \otimes \mathrm{End}(E^G)$.
	\end{itemize}
	\end{theorem}
	
	\begin{proof}

	By choosing a fixed set of constant Dirac matrices, $\spin(r,s)$ has a corresponding unique  representation on  $\com^{d_\gamma} $. Thus: 
	\begin{equation} \spin(r,s)  \subset \glc , \qquad \spin(r,s)^\com  \subset \glc . \label{eq:spin_sub}\end{equation}
	(Different choices of those matrices correspond to a fixed  similarity transformation and potentially, in the odd dimensional case, an additional sign).

	We now show that we may extend the symmetry on the vector bundle $E$ from $\spin(r,s)$ or $\spin(r,s)^\com $ to  a $ \glc$ vector bundle $E^G$ (same set, but larger structure group) such that $\gamma$ is defined in  $T^*M \otimes \mathrm{End}(E^G )$. 
	
	The existence of $E^G$ is obvious, since the vector bundle may be defined by an open cover $\{V_\alpha\}$ of subsets of $M$ and transition functions $\Psi_{\alpha \beta}: V_\alpha \cap V_\beta \rightarrow \spin(r,s) $ (or  $\Psi_{\alpha \beta}: V_\alpha \cap V_\beta \rightarrow \spin(r,s)^\com $  in the $\spin^\com$ case) such that the cocycle condition is fulfilled: $\Psi_{\alpha \beta} \cdot \Psi_{ \beta \gamma} \cdot \Psi_{\gamma \alpha} =  \openone$ on $V_\alpha \cap V_\beta \cap V_\gamma$. However, due to Eq.~(\ref{eq:spin_sub}), those are valid transition functions for a $GL(d_\gamma,\com )$ bundle as well. (Essentially, one simply ``forgets" the spin structure and thus allows more general base transformations). $T^*M$ may be considered a $\glr$ vector bundle in complete analogy by admitting arbitrary frames instead of only orthonormal frames in the frame bundle.
	Now forming the bundle   $T^*M \otimes \mathrm{End}(E^G) $, this obviously transforms under  $\glr  \times  \glc$  by construction. 
	
	Since $E^R$ and $E$ are identical as sets and the identification is linear in the fibres, $\gamma$ is obviously well-defined as a section of $T^*M \otimes \mathrm{End}(E^G) $. 
	\end{proof}

	All this is essentially very basic, the only subtlety being the fact that in order to define $\gamma$ as a global section, one needs to transform space time index and spinor index at the same time with the same element of the $\spin$ or $\spin ^\com$ group as shown above. However, as soon as $\gamma$ is defined, we may consider more general transformations transforming the $T^*M$ part and the $\mathrm{End}(E)$ part separately, but consistently with the bundle structure. Obviously, with those more general transformations, the coordinate expression of $\gamma$ is no longer a set of constant Dirac matrices.

\section{Reduction of structure group for arbitrary globally defined Dirac structure  \label{sec:reduction}}

	In this section, we show the central result of this paper, namely that the existence of a Dirac structure implies a $\spin^\com$ structure:
	
	\begin{theorem} For any spin base invariant Dirac structure on an orientable pseudo-Riemannian manifold $(M,g)$ i.e. an arbitrary complex vector bundle $E$ over $M$ of dimension $d_\gamma$ with a global section $\gamma \in \Gamma(T^*M \otimes \mathrm{End}(E))$ such that the Clifford condition  (\ref{eq:dirac_str}) is fulfilled, there is a $\spin^\com(r,s)$- structure such that $\gamma$ is the Dirac structure corresponding to this $\spin^\com(r,s)$- structure according to theorem \ref{thm:ext}. In particular, a Dirac structure exists iff a $\spin^\com(r,s)$- structure exists, i.e. the same topological obstructions apply.
	\end{theorem}
	
	So in a more physical language, we may partially gauge fix the large gauge group $\glr \times  \glc $ to a smaller  group  $\spin^\com(r,s)$, reducing the whole structure to the $\spin^\com$ case.
	
	\begin{remark}
	 It is interesting to note that the less intuitive group  $\spin^\com(r,s)$ and not  $\spin(r,s)$ appears naturally for spin base invariant structures, so this may be considered as a more natural setting for deriving this $\spin^\com(r,s)$  structure instead of the normal approach of ``artificially'' adding a $\uone$ gauge symmetry to alleviate the obstruction towards existence of $\spin$ structures.
	\end{remark}
	
	\begin{proof}

	Let $\gamma$ be a Dirac structure as defined above.
	For two arbitrary bundles $B_1, B_2$ over $M$ we define $ B_1 \times_M B_2 $ as the pullback of the direct product $B_1 \times B_2$ via the diagonal map $d: M \rightarrow M \times M$, i.e. we only consider those points in the direct product which have the same base point in $M$, thus again forming a bundle over $M$ (and not over $M \times M$. This is the general bundle analog of the construction of Whitney sum bundles for vector bundles).
	
	For the proof, we first consider the case of even dimensional $M$:
	
	Considering $E$ as a general complex vector bundle, we consider the bundle 
	\begin{equation} P = L(M) \times_M L^\com (E),\end{equation}
	where $L(M)$ denotes linear frames over $M$, $L^\com (E)$ the complex linear frames of $E$. $P$ is a principal  bundle over $M$ with structure group $\glr \times  \glc $. It consists of all pairs of real linear frames on  $TM$ and of complex linear frames on $E$ over the same base point in $M$. 
	
	 Using orientability and the (pseudo)-Riemannian metric $g$ we may restrict the frames in $L(M)$ to positively oriented, orthonormal ones, leading to a principal bundle
	\begin{equation} \tilde{P} = SO(M) \times_M L^\com (E),\end{equation}
	with structure group $SO(r,s) \times \glc$.
	
	Any point $p \in \tilde{P}$ defines a base of $T_xM$ and a base of $E_x$ for $x=\pi(p) \in M$, and hence a map $p^*: T_x^*M \otimes \mathrm{End}(E_x) \rightarrow (\real^D)^* \otimes \com({d_\gamma}) $ where $\com(n)$ denotes the set of complex $n\times n$ matrices. 
	
	We choose a fixed ``reference'' set $\gamma^f_i \in \com(d_\gamma), i = 1,...,D $ of  $\gamma$ matrices fulfilling Eq.~(\ref{eq:flatdirac_def}) and define 
	\begin{equation}
        P^f \coloneqq \{p\in \tilde{P} | p^*(\gamma)(e_i) = \gamma^f_i \} , 
    \end{equation}
    where $e_i =(0,..,0,1,0,..)$ form the standard base of $\real ^D$, i.e. we restrict to those frames for which the Dirac structure $\gamma$ expressed in the bases defined by $p$ coincides with $\gamma^f_i$. (This is analogous to reducing the structure group of a Riemannian manifold to $O(N)$ by considering only frames which are orthogonal, i.e. where the metric $g$ expressed in the base is a standard metric $\textrm{diag}(1,..1,-1,...,-1)$).
	The choice of $\gamma^f_i $ defines in addition a representation of $\spin^\com(r,s) $ on $\com^{d_\gamma}$ and thus and  embedding of $\spin^\com $ as a subgroup of $\glc$. 
	
	We now  show that  a $\spin^\com$ structure may be constructed by taking the quotient $P^f / \real^+ $ where the multiplicative group $\real^+$ acts by scaling the frames in $E$ only. The reason for this quotient is the fact that the Dirac structure does not define a notion of scale on $E$ since an overall scale drops out in the transformation of $\gamma$. 
	
	First, we show that $P^f$ itself forms a  $\spin(r,s)^\com \times \real^+ $ principal bundle. 
	We identify $\spin(r,s)^\com $   with a subgroup of $SO(r,s) \times \glc$ by $g \mapsto (\tilde{g},g)$ where $\tilde{g}$ is the image of $g$  in $\so(r,s)$ in $ Eq.~(\ref{eq:def_spinc})$. It is important for the following that $\spin(r,s)^\com $ acts on both factors at the same time.
	
	Let $F_x = \pi^{-1}(x) \subset P^f$ be the fibre of $P^f$ over the base point $x$. We show that $\spin(r,s)^\com \times \real ^+$ acts via a free and transitive action  from the right on   $F_x$:
	\begin{itemize}
		\item \emph{${F}_x$ is not empty}:  This follows since for  a \mbox{(pseudo-)}orthornomal frame, Eq.~(\ref{eq_cliff}) reduces to Eq.~(\ref{eq_flatdirac_trafo}) and in even dimensions all representations of the Clifford algebra are equivalent by a similarity transformation.
	    \item \emph{$\spin(r,s)^\com \times \real^+ $ action is well-defined:} Since in $p^*(\gamma)(e_i)$ the index $i$ transforms as a co-vector under  a Lorentz transformation and the endomorphism part by conjugation  with $Spin(r,s)^\com $, it follows from the transformation properties  of the $\gamma$ matrices Eq.~(\ref{eq:inv_gamma}) that the image of $ p \cdot g $ of $p \in F_x$ under $g \in   \spin(r,s)^\com $ is again in $F_x$, hence the action of the subgroup is well-defined on $P^f$.
	    \item \emph{The action is transitive:}
  	      By construction, the action of $SO(r,s) \times \glc$ is transitive on the fibres of $\tilde{P}$. Hence, if $p,q \in F_x$ are two points in the same fibre $F_x$, there are $(g,h) \in   SO(r,s) \times \glc$ such that $q = p \cdot (g,h)$. We have to show that $(g,h)$ may be chosen such that $(g,h) \in \spin(r,s)^\com \times \real^+.$  Let $k \in  \spin^\com(r,s)$  be such that its corresponding Lorentz transformation component of $\rho(k)$ is $g$ under Eq.~(\ref{eq:def_spinc}). Then, $q \cdot k^{-1} \in F_x$ as well by the statement in the previous item, hence $ r \coloneqq p \cdot (g,h) \cdot  k^{-1} = p \cdot (g g^{-1} , h k^{-1}) = p  \cdot (e , h k^{-1}) \in F_x$ where $e$ denotes the identity in $\so(r,s)$.  Hence, for all $i $: 
	      \begin{equation} 
	         ~~~~ p^*(\gamma)(e_i) = \gamma^f_i  =  r^*(\gamma)(e_i) =  h k^{-1}\gamma^f_i(h k^{-1})^{-1},
	      \end{equation}
	      so the $\gamma^f_i $ are invariant under conjugation with $h k^{-1}$. Since the identity matrix  $\openone$  and products of the $\gamma^f_i $ span the whole set of complex matrices in $\com(d_\gamma)$,  $h k^{-1}$ must be in the center of $\glc$ and hence a multiple of $\openone$: $h k^{-1} = \lambda \openone $ with $ \lambda \in \com^* = \com \setminus \{0\}$. Hence, writing $\lambda = |\lambda| u$ with $u \in \uone $,  $(g,h) = k \cdot \lambda \openone = k \cdot  u  |\lambda|  \openone   $ maybe identified with an element of $ \spin(r,s)^\com \times \real^+ $ since $(-1,-1) \in \spin(r,s) \times \uone $ acts trivially on $P^f$.   
	          
	    \item \emph{Free action:} The fact that the action is free may be seen easily: If $ p \cdot g =p $ for $(g, \lambda)  \in   \spin(r,s)^\com \times \real^+ $ then obviously  $\lambda=1$ since it acts by scaling on the frame  in  $L^\com (E)$. Furthermore, the image of $g$ in $\glc$ must be the identity. This implies
	    	    that $g$ must be itself the identity, since the kernel of the map $\spin(r,s) \times  \uone  \rightarrow \glc$ is exactly $\mathds{Z}_2 = \{(1,1),(-1,-1)\}$ which is divided out in the definition of $\spin^C$. (This is a main reason to consider $\spin^\com$ instead of $\spin(r,s) \times  \uone$.)   
	\end{itemize}

    So $P^f$ indeed has the structure of a $\spin(r,s)^\com \times \real ^+$ principal bundle. We define
    \begin{equation} \label{eq:constr_spinc}
     \spin^\com(M)  \coloneqq P^f / \real^+ , 
    \end{equation}
    where $\spin^\com(M)$ for the moment is just a name for the quotient, but we will show that this is indeed a $\spin^\com$ bundle:
    The group   $\real ^+$ acts trivially on the quotient by construction. Going through the previous steps one easily checks that $ \spin(r,s)^\com $ acts freely and transitively on this quotient, so we have a  $ \spin(r,s)^\com $ principal bundle. The fact that this is indeed a smooth bundle follows from the fact that $P^f$ is a smooth bundle by its definition which allows smooth local trivializations by the inverse function theorem, and the fact that the group action of $\real^+$ on $P^f$ is free and smooth.

    To finalize the proof that this is indeed a $\spin^\com $ structure we have to define a $\uone$  bundle over $M$ such that there is an equivariant mapping $\Phi^\com$ as in Eq.~(\ref{eq;spin_com_cov}). The $U(1)$ bundle may be easily constructed as an associated bundle:
    \begin{equation} 
    U_1(M) \coloneqq  (\spin^\com(M) \times \uone)/ \sim , 
    \end{equation}
    where $(p,h) \sim  (p \cdot [(g,u)]^{-1}, u^2 h)$ for arbitrary $[(g,u)] \in \spin^\com $ and $(p,h) \in (\spin^\com(M) \times \uone)$. Since $\uone$ acts on $U_1(M)$ from the right, this is indeed a $\uone$-principal bundle and an equivariant  mapping from $\spin^\com(M)$ to $U_1(M)$ is defined by: $\Phi^\com_{\uone} (p) \coloneqq (p,e)/\sim$, where $e$ is the identity in $\uone$.
    
    Now, there is an obvious equivariant mapping from $\spin^\com(M)$ to $\so(M)$, just mapping $p \in \spin^\com(M)$ to the  $\so(M)$ component of one of its representatives in the quotient modulo $\real^+$. Since $\real^+$ only acts on the second $L^\com (E)$-component, this is well-defined independently of the choice of the representative. 
    
    Finally, combining both maps, we get a well-defined equivariant map $\Phi^\com: \spin^\com(M) \rightarrow \so(M) \times U_1(M)$, so $\spin^\com(M) $ indeed defines a $\spin^\com $ structure on $M$, and the proof is finished for the even dimensional case.

	The case of an odd dimensional manifold is only slightly more complex: The only additional complexity is the fact that there are two inequivalent representations of the Clifford algebra $\cl(r,s)$ which may be mapped to each other by $\gamma^f_i \mapsto -  \gamma ^f_i$. So in order to do the whole construction above, we have to choose the right set of constant Dirac matrices. However, assuming that $M$ is connected, this ``right set'' may not change with the point $x\in M$ chosen due to continuity of the Dirac structure $\gamma$. Hence, everything works out completely analogously.  
		\end{proof}

    \begin{remark} ~
    \begin{itemize}
        \item
        According to Eq.~(\ref{eq:constr_spinc}), $P^f$ may be considered an $\real^+$ principal bundle over  $\spin^\com(M)$. However, since $\real^+$ is contractible, any $\real^+$ principal bundle is trivial     \cite{hatcherAlgebraicTopology2002}, hence $P^f$ may be identified (in a non-canonical way) with $ \spin^\com(M) \times \real^+$. Hence, $ \spin^\com(M)$ may be embedded as a submanifold of  $P^f$ (For an elementary, but less general and less elegant  proof of triviality of $\real^+$ bundles not using arguments from algebraic topology, see appendix \ref{sec_elem_triv}) . However, this embedding is non-canonical essentially due to the lack of the notion of a scale on the bundle $E$, which would allow to gauge-fix the $\real^+$ action. A scale may be introduced by introducing a metric or ''pseudo-hermitian''  structure on $E$, i.e.  a non-degenerate but generally indefinite scalar product on the fibres. However, this additional structure is not needed for the proof.
        \item This lack of scale is the reason why we do not restrict the group to $\slc$ instead of $\glc$ as used e.g. in     Ref.~\onlinecite{giesFermionsGravityLocal2014}.
    \end{itemize}
    \end{remark}

\section{Spin metric and connection} \label{sec:metric_conn}
    In order to define a spin metric and a compatible connection, we make one additional assumption on the manifold $(M,g)$ for the non-Riemannian case ($ r\neq 0 $ and $s\neq0$): We assume that $M$ is space and time-orientable, i.e. there is a global orientation of ''space'' and (potentially multi-dimensional) ``time'' sub-frames separately, i.e., for any  pseudo-orthonormal  frame  $f=(v_1,\ldots,v_r, v_r+1,\ldots v_{r+s}) $ with $g(v_i,v_i) = +1 $ for $i=1,\ldots,r$ and $g(v_i,v_i) = -1 $ for $i=r+1,\ldots,r+s$, there is a consistent assignment of ``positive'' or ``negative'' orientation in $ (v_1,\ldots,v_r) $ and $( v_{r+1},\ldots v_{r+s})$ separately. Obviously, orientation and ``time'' orientation imply ``space'' orientation and vice versa. In the special case $r=1$ or $s=1$, this coincides with the standard notation of an orientable,  time orientable Lorentz manifold. 
    
    With this additional structure, we may reduce frames in $\so(M)$ to only those which are both positively ``time'' and ``space''-oriented. In this way, we reduce the structure group from $\so(r,s)$ to $\so_0(r,s)$, the identity component of $\so(r,s)$. This may now be transferred to the situation of a manifold with Dirac structure:
    By restricting in the construction of $P^f$ the space-time frames to only those which are both positively ``time'' and ``space''-oriented, yielding a $\spin^\com_0(r,s)\times \real^+$ bundle  $P^f_0$ where $\spin^\com_0(r,s)$ is the component of the identity of $\spin^\com(r,s)$.
    
    Now, by choosing a section in the $\real^+$ bundle, we get a $\spin^\com_0(r,s)$ principal bundle. Its associated vector bundle $E$ now inherits an --  up to a global scale --  unique metric by pulling the (up to a scale) unique metric on $V=\com^{d_\gamma}$  $<,>_{\scriptscriptstyle V}$ with 
     \begin{equation} \label{eq:cond_metric_main}
         \begin{split}
             <\psi, \psi>_{\scriptscriptstyle V} ^* & =  <\psi,\psi>_{\scriptscriptstyle V}^{~}, \\  \forall v \in V \subset \cl(V): ~ <\gamma(v) \phi,  \psi>_{\scriptscriptstyle V} & = (-1)^s <\psi,\gamma(v) \phi>_{\scriptscriptstyle V}  
        \end{split}
     \end{equation}  
    (see appendix \ref{sec:metr_com_vector}) by choosing for a vector $ X \in E$ a representative $(p,x)$ of the $\sim$ equivalence class $[(p,x)] = (p,x) / \sim$ in the $\spin^\com_0(r,s)$ analog of Eq.~(\ref{eq:assoc_vb_sp}) and setting $<X,X> \coloneqq <x,x>_{\scriptscriptstyle V}$. This is well-defined, since the metric $<,>_{\scriptscriptstyle V} $ is $\spin^\com_0(r,s)$ invariant. By the standard polarization trick, we may define from this $<\psi,\phi>$ for arbitrary sections $\psi,\phi$ in $E$, which  fulfills the obvious bundle analog of equation Eq.~(\ref{eq:cond_metric_main}). 
    
    \begin{remark} As soon as $<,>$ is defined $f\cdot <,>$ is a valid metric for any $f: M\rightarrow\real^+$. However, given $<,>$ we may use it to ``fix the scale'' (up to one global scale factor in  $  \real^+$ in the previous paragraph. In this sense,  choosing an embedding of $ \spin_0^\com(M)$ in $P^f_0$ corresponds to choosing a scale and thus fixing $f$. 
    \end{remark}

    Once $<,>$ is defined, we may extend the structure group as in section \ref{sec:ext_spin_dirac}, thus obtaining a metric on $E$ as vector bundle with structure group $\glc$. (Obviously, whenever more general local trivializations of $E$ are used, the local expression of $<,>$ is no longer $<,>_{\scriptscriptstyle V}$ but has to be transformed consistently with the trivialization).
    
    To define a compatible covariant derivative or equivalently, a compatible connection, we start with the Levi-Civita connection on $M$. 
    
    As is well known, we may identify a connection with a globally defined, equivariant Lie-algebra valued one-form on $\so_0(M)$. (In the following,  we will deliberately use the different, but equivalent views on connections on vector and principal bundles to simplify the derivation. They are  briefly summarized  in appendix \ref{sec:conn_principal}.)
    
    This one-form can be pulled back from $\so_0(M)$ with the double-cover projection to  $\spin^\com_0(M)$, defining a connection on the   $\spin^\com_0(r,s)$ principal bundle  $\spin^\com_0(M)$, and thus a covariant derivative on $E$. Since the defining  properties of a covariant derivative on a vector bundle do not relate to the structure group, but only to the vector space structure of each fibre (and the $\mathcal{F}(M)$-module structure of $E$), this defines a valid covariant derivative if we consider $E$ as a bundle with larger structure group $\glc$ as well. 
    
    Now, compatibility of this connection is obvious without calculation: 
    
    For any associated vector bundle, the covariant derivative  $\nabla_X s(x) $ of a section $s$  may be obtained by choosing a curve $\chi$ in $M$, starting at $x$ and tangent to $X(x)$, lifting it to a horizontal curve $\tilde{\chi}$  in the frame bundle  $\spin^\com_0(M)$ and then taking the ordinary derivative of the components of $s$ in the base defined by $\tilde{\chi}$. 
    
    Now, for the Dirac structure $\gamma$ the coordinate expression using $\tilde{\chi}$   of $\gamma$ when considered as a section in $T^*(M) \times \mathrm{End}(E)$ is constant as shown in sections \ref{sec:exist_dirac_spin} and \ref{sec:reduction} , thus $\nabla \gamma=0$.
    Similarly, if $<,>$ is considered a section of $E^*\otimes_M E^*$, its coordinate expression using  $\tilde{\chi}$ is constant by the construction above. Thus, we may conclude:
    \[  \nabla \gamma =0  , \text{ and } \nabla < , > =0, \] 
    where $\nabla$ acts as the natural tensorial extension on the respective bundles. (The second condition,  using the identification of $<,>$ with a  section of a bundle, might look a bit strange, a more intuitive representation is 
    \[\nabla <\psi,\phi> =  <\nabla \psi,\phi> + <\psi,\nabla \phi> \]  for all  $\phi,\psi$.)

    \begin{remark} Despite looking very similar,  the ``ontological'' status of the compatibility condition $\nabla \gamma =0$  in the local spin invariance setting is quite different form the vielbein postulate: As already stated in section \ref{sec:spin_inv_geom}, the ``vielbein postulate'' is just the condition $\nabla id =0$, where $id$ is the identity on the vector bundle $TM$. It is automatically trivially fulfilled  for any connection and only allows - as  a computational tool -  to translate expressions in holonomic, coordinate indices, to those in anholonomic vielbein indices. If we extend the structure group from $\so(r,s)$ to $\glr$, the corresponding vielbeins are no longer orthogonal, the vielbein postulate formulated as $\nabla id =0$ still trivially holds, but does not restrict the connection to be metric compatible. An expression $\nabla (e_i) =0$ for only orthogonal vielbeins  does not make any geometrical sense in this setting since the vielbeins are arbitrary, only locally defined sections in the frame bundle - the geometric meaningful condition being the standard Levi-Civita condition $\nabla g=0$.
    
    For the condition $\nabla \gamma=0$, a similar statement seems to hold at first sight, since it is automatically fulfilled in a $\spin^\com_0(r,s)$ setting. However, as soon as we extend the structure group to $\glc$, the crucial difference between this condition and the ``vielbein postulate'' becomes obvious: Since the Dirac structure is a global object, the condition $\nabla \gamma=0$ is perfectly well-defined in this setting and it automatically ensures that $\nabla g=0$ since $g$ may be expressed by $\gamma$. 
    \end{remark}
         
    Starting from the covariant derivative $\nabla$ on $E$, one can study more general connections and formulate suitable compatibility conditions, as done e.g. in local coordinates in   Ref.~\onlinecite{giesFermionsGravityLocal2014}. The difference between such a covariant derivative and $\nabla$ has been denominated ``spin torsion''.   We will not study this here  in more detail and only hint that this allows for additional degrees of freedom in the formulation of a theory of gravity.

\section{An example: The sphere $S^n$} \label{sec:sphere}

	In the paper by Lippoldt and Gies    (Ref.~\onlinecite{giesGlobalSurplusesSpinbase2015}) it is claimed that using suitable spin base transformations, a global trivialization may be achieved for $S^2$ despite the fact that $S^2$ is not parallelizable. However, since smoothness is only shown for $\gamma$ and the eigenvectors of the Dirac operator, it is not quite clear what is meant by this statement from a  geometric perspective, in particular since, as shown above,  the Dirac structure is always globally defined in case of a $\spin$ structure, which is known (as a $\spin(2,0)$ structure, not $\spin(1,1)$) to exist for $S^2$, and hence its pullback via any smooth coordinates and local sections is always smooth.
	
	We show here, as a generalization to arbitrary dimensions, that a stronger statement holds, namely the spinor bundle, when considered as a bundle over the larger spin base transformation group, is trivial for any $S^n$. (This does obviously not hold when only considering $\spin(n)$ transformations since otherwise the tangent bundle would be trivial as well in contradiction to the hairy ball theorem).

    First, we assume that the dimension of the sphere is even, since the argument is a bit more obvious in this case, and then extend to the odd dimensional case. We are considering the case of spin over a Riemannian manifold, since $S^{2n}$ does not carry any Lorentz metric.)

    We use the following facts:
    \begin{itemize}
    	\item Denoting by $N(S^{2n} )$ the normal bundle of $TS^{2n}$, the sum bundle $TS^{2n} \oplus N(S^{2n} )$ is the trivial $\mathds{R}^{2n+1}$ vector bundle  obtained by embedding $S^{2n}$ in  $\mathds{R}^{2n+1}$ and pulling back the trivial bundle $T\mathds{R}^{2n+1} = \mathds{R}^{2n+1}\oplus \mathds{R}^{2n+1}$ to $S^{2n}$.
    	\item The dimensions of the spinor bundles in $2n$ and $2n+1$ dimensions are the same.
    	\item 	$\spin(2n+1 ) \subset SU(d_\gamma)$.
    	\item $S^{2n}$ being simply connected, the $\spin$ structure  on $S^{2n}$ is unique up to isomorphisms.  	
    \end{itemize}	
    The first item is obvious from the definition of a normal bundle, the second is a well-known fact about Clifford algebras, the third item is well known for spinor representations over a Riemannian manifold (no spin metric including $\gamma_0$ needed). 

    We now consider the (necessarily trivial) $ \spin(2n+1)$ spinor bundle over $\mathds{R}^{2n+1}$  and denote its pullback to $S^{2n}$  via the inclusion map as $\tilde{E} = S^{2n} \times {\com^{d_\gamma}}$ and the corresponding trivial principal bundle $\tilde{P} = S^{2n} \times \spin(2n+1)$. 
    
    Denoting by $\Pi$ the double covering map: $\Pi: \spin(2n+1) \rightarrow \so(2n+1)$, we consider the subbundle $P\subset \tilde{P}$: 
    \begin{equation}
     P\coloneqq   \{(x,g) \in S^{2n} \times \spin(2n+1)|\forall X \in N_x(S^{2n}):  \Pi(g) X =X \} ~,
    \end{equation}
    i.e., we restrict the Spin transformations to those whose corresponding $SO(2n+1)$-projections under double covering leave the normal and hence the tangent spaces invariant.
    
    Now, by construction, the fibres of $P$ are double covers of the frame bundle of $S^{2n}$ and diffeomorphic to $\spin(2n)$, hence $P$ is the (unique) spin bundle over $S^{2n}$. Since $P$ is a sub-bundle of $\tilde{P}$ we may consider the spinor bundle $\tilde{E}$ as as $\spin(2n)$ bundle $E$ by reduction of structure group from $\spin(2n+1)$ to $\spin(2n)$ (so both are the same as vector bundles).   
    
    Since $\spin(2n+1)$ acts on $\tilde{E} \coloneqq E$, its action is a subset  of the spin base transformations. Hence, $E$ is non-trivial as a $\spin(2n)$ bundle, but allowing for the more general spin base transformations, it may actually be trivialized, and we are done for the even dimensional case.
    
    For the case of odd dimensional spheres $S^{2n+1}$, the only difference is  that the representations of $\spin(2n+1)$ and $\spin(2n+2)$ resulting from the respective Clifford algebra representations no longer have the same dimension, the dimension  for $\spin(2n+2)$ being twice that for $\spin(2n+1)$, and that there are two inequivalent Clifford representations in the odd dimensional case. However, for the representations of the $\spin$ groups, which only consist of products of even elements in the Clifford algebra, the even dimensional representation is well known to be reducible (to the eigenspaces of $\gamma_*$ for eigenvalues $\pm 1$). 
    
    Hence, we may restrict the bundle construction to one handedness in the chiral representation. Thus, the dimension of the representations do fit and we may use the same construction as for the even dimensional case. Since the spin structure over $S^{2n+1}$ is known to be unique, it follows that the whole construction does not depend on the choice involved.  \hfill   $\square$

\section{Conclusion and Outlook} \label{sec:conclusion}
    In this paper we have shown that a global differential-geometric formulation of the spin base formalism is not only possible, but beneficiary. It facilitates to answer global questions like possible obstructions and connections to other formalisms and  simplifies some proofs which are rather tedious in a local coordinate formulation. 
    
    The differential-geometric view on local spin base invariance shows that the Dirac structure is indeed a global object existing whenever spin may be defined. This is a strong hint that the Dirac structure may indeed be a better variable than vielbeins, which exist only as local sections of a frame bundle, but have no global meaning. Furthermore, it turns out that the argument ``fermions exist in nature'' is rather an argument in favor of Dirac structures, and not in favor of vielbeins as variables for a quantum gravity theory with fermions. (Though, obviously, the final decision must be made on the basis of experiments and not of aesthetics and mathematical simplicity and elegance.)
    
    The fact that from a Dirac structure, we can always construct a $\spin^\com$ structure implies that the same topological obstructions apply to both approaches. Nevertheless, the spin base approach is more general as it allows for additional degrees of freedom like spin torsion and extensions of standard general relativity. 
    
    Apart from considering more general theories, it might be worthwhile to study conventional approaches to quantizing gravity in those variables.  In particular, besides studying it from a path integral and functional renormalization group view, it might be worthwhile to study the structure of the constraints in a Hamiltonian/Dirac formalism approach using those variables and to check whether it may be beneficial to formulate approaches like loop quantum gravity and spin foams in this language.

\section*{Acknowledgment}
    I would like to thank Holger Gies  for valuable feedback on the manuscript.

\section*{Data Availability Statement}
    The data that support the findings of this study are available within the article.

\appendix
\section{Some properties of Clifford algebras\label{sec:app_prop_cliff}}   
    We collect here some well known facts about Clifford algebras (see e.g.   Refs.~\onlinecite{lawsonSpinGeometryPMS382016a,figueroa-ofarrillSpinGeometry2017,cornwellGroupTheoryPhysics1989}):

    \begin{itemize}
        \item For a real vector  space $V$ with possibly indefinite, but non-degenerate scalar product $<,>$ there is a naturally defined real Clifford algebra $\cl(V)$ which is obtained by taking the tensor algebra $\mathcal{T}(V) = \mathds{R}+  V + V \otimes V   + \ldots$  and dividing it by the ideal $\mathcal{I}$ generated by elements $ x \otimes x + <x,x> \cdot 1$. This construction is functorial and does not require the choice of any base of $V$. By construction, elements of $V$ may be naturally identified with a subset of $\cl(V) $  due to the fact that  there are no elements of order $1$ in the ideal $\mathcal{I}_x$. For those elements, the Clifford relation:  $ x \cdot x = - <x,x>$ or (by a polarization argument) $\{x,y\} = -2 <x,y>$ holds by construction. 
        \item The Clifford algebra is a $\mathds{Z}_2$-graded algebra and a $\mathds{Z}$-filtered algebra. As a vector space, it may be identified with $\Lambda(V)$, the exterior algebra. This is obviously not an algebra isomorphism.
        \item The Clifford algebras have been completely classified, and, depending on the signature $(r,s)$ of $<,>$, they are of the form $\real(n), \com(n), \mathds{H}(n), \real(n)\oplus\real(n) , \mathds{H}(n)\oplus\mathds{H}(n) $ for $n$ some power of $2$, where $\mathds{K}(n)$ is the space of $n\times n$ matrices over $\mathds{K}$, and $\mathds{H}$ denotes the quaternions. 
        \item The group $\pin(r,s)$ is defined as the product of arbitrary elements of $v_i\in V$  of ``length squared'' $\pm 1$ in $\cl(V): g= v_1\cdot .... \cdot v_k$ with $<v_i,v_i> = \pm 1$. 
        \item The double covering of $O(r,s)$ by $\pin(r,s)$ is defined by $g \mapsto L $ with 
            \begin{equation} \label{eq:def_cov}  \widetilde{\textrm{Ad}}_g(\hat{x}) =  \widehat{L x}, \end{equation}
            where $\hat{x}$ denotes the canonical image of $x\in V$ in $\cl(V)$, (in a local basis: $x^\mu \gamma_\mu$) and  $\widetilde{\textrm{Ad}}_g = \pm {\textrm{Ad}}_g$ depending on whether $g$ is an  even or odd product of elements of $V$. It can be shown that $\widetilde{\textrm{Ad}}_v$ for $v\in V$ with $<v,v>=\pm 1$ corresponds under this mapping to a reflection along the hypersurface vertical to $v$. Since all elements in
            $\textrm{O}(r,s)$ may be written as a finite product of such reflections, it follows that this mapping is indeed onto, has kernel ${\pm 1}$ and thus indeed defines a double cover.  
            
        \item $\spin(r,s)$ consists of all those elements of     $\pin(r,s)$ which are an even product of $v_i\in V$. By the same construction, it is a double cover of $\so(r,s)$.
        
        \item For $\spin(r,s)$ all elements are even by definition and hence  $\widetilde{\textrm{Ad}}_v =  {\textrm{Ad}}_v$ for $v \in \spin(r,s)$.
        
        \item Complex irreducible representations are of dimension $2^{d_\gamma}$ with $d_\gamma = \lfloor \dim(V)/2 \rfloor$. 
        
        \item  In the case of $\dim(V)$ even, there is  up to equivalence (realizable by adjoining with a $\glc$ element) only one irreducible representation of the Clifford algebra. As a representation of the Spin group it is reducible and decomposes into two irreducible representations, the Weyl representations,  corresponding to eigenvalues of $\pm 1$ of $\gamma^* \coloneqq i^p \gamma_1  \ldots \gamma_{\dim(V)}$  ($p$ is chosen depending on $r,s$ such that $\gamma^*$ is self-adjoint).
        
        \item  In the case of $\dim(V)$ odd, there are  up to equivalence  two such irreducible representations of the Clifford algebra. They can be transformed into each other by replacing $\gamma$ by $- \gamma$.  As a representation of the Spin-group both representations are irreducible.
    \end{itemize}

    Choosing a base $e_\mu$ of $V$  with  corresponding $\gamma_\mu \in \cl(V)$, we conclude from Eq.~(\ref{eq:def_cov}):
    \begin{equation}
        \begin{split}
           \textrm{Ad}_g ( x^\mu  \gamma_\mu)  & =   L^\nu_{~~\mu} x^\mu \gamma_\nu      \\
           \Rightarrow ~    x^\mu  ~  ( g \gamma_\mu g^{-1}) & =    x^\mu L^\nu_{~~\mu}  \gamma_\nu \\  
           \Rightarrow ~ g \gamma_\mu g^{-1} & =     L^\nu_{~~\mu}  \gamma_\nu , 
        \end{split}
    \end{equation}
    and hence
    \begin{equation}     
         (L^{-1})^\mu_{~~\rho} g \gamma_\mu g^{-1}  =     \gamma_\rho ~. \label{eq:inv_gamma}
     \end{equation}
     Note that under a Lorentz transformation a co-vector $x_\mu$ transforms as 
     $ x''_\mu = (L^{-1})^\mu_{~~\rho} x_\mu$, so Eq.~(\ref{eq:inv_gamma}) states that $\gamma$ is invariant under a $\spin(r,s)$ transformations if the space-time index is transformed as a co-vector under the corresponding Lorentz transformation and at the same time the ``endomorphism part'' by conjugation, as induced by the representation of $\spin(r,s)$ on a complex vector space.

     \section{Metric on $\com^{d_\gamma}$} \label{sec:metr_com_vector}
     
    We now show the existence and - up to a scale - uniqueness of a ``metric'', i.e. an non-degenerate, but possibly indefinite sesqui-linear form on $\com^{d_\gamma}$ with the property:
     \begin{equation} \label{eq:cond_metric}
        \begin{split}
             <\psi, \psi> ^* & =  <\psi,\psi>\\ 
             \forall v \in V \subset \cl(V): ~ <\gamma(v) \phi,  \psi> ^* &= (-1)^s <\psi,\gamma(v) \phi>  ~.
        \end{split}
     \end{equation}  
     
    From those properties, it follows in particular that the metric is invariant under $\spin_0(r,s)$, the identity component of $\spin(r,s)$:
    
    \begin{itemize}
        \item For the Riemannian case ($r=0$ or $s=0$), $\spin(r,0) \cong \spin(0,r)$ is simply connected. Since $ \spin_0(r,0)= \spin(r,0)$  may be represented  by even products of $v_i$  with $<v_i,v_i>= +1$, a possible  minus sign in Eq.~(\ref{eq:cond_metric}) cancels when computing $<g \phi, g\psi> $ for $g = v_1 \ldots v_{2k}$.
        
    \item For both $r,s\neq 0$, $\spin(r,s)$ has two components, which are mapped to each other by a combination of (in physics language) ``time reversal'' and ``parity''. (for both $r,s\leq2$ ``time'' here  has to be understood as generalized, multidimensional mathematical time,)  $\spin_0(r,s)$ may hence be identified with  products of $v_i$  with $<v_i,v_i>= \pm 1$, where both the number of $v_i$ with $<v_i,v_i>= + 1$ and the number of $v_i$ with $<v_i,v_i>= - 1$ are even. Hence, the minus signs cancel again in $<g \phi, g\psi> $, and the metric is invariant under $\spin_0(r,s)$ (but not under  $\spin(r,s) $).
    \end{itemize}

     Those properties are needed to define a physical theory with a real Lagrangian  (see Refs.~\onlinecite{giesFermionsGravityLocal2014,lippoldtSpinbaseInvarianceFermions2015}) If one considers the fermionic fields as anti-commuting Grassmann variables as needed e.g. for a path integral approach, the additional minus sign in the reality condition may be accommodated for by a factor $i$ in the metric. For the purpose of this paper, we consider the fermionic fields simply as sections in a vector bundle without additional algebraic properties.

     \textit{\textbf{Existence:}} We use the fact that there is an explicit representation of the Dirac matrices in arbitrary dimension as a tensor product of Pauli matrics $\sigma_i$, see e.g.  Ref.~\onlinecite{cornwellGroupTheoryPhysics1989}. We start with such a representation for the Riemannian case. Here, all Dirac matrices are hermitian matrices in this representation. (Since Ref.~\onlinecite{cornwellGroupTheoryPhysics1989} deals with the Lorentzian case, we have to replace $\gamma_0$ by $\gamma_D\coloneqq i \gamma_0$, which is hermitian as well, to get the Riemannian representation). Hence, for the Riemannian case we simply set
    \[ <\psi,\phi> \coloneqq  \psi^\dagger \cdot \phi, \]
    which obviously fulfills Eq.~(\ref{eq:cond_metric}).

    We denote the chosen  Riemannian $\gamma$ matrices corresponding to an orthonormal  base by $\gamma^R_i$.
    Then, for the general case of signature $(r,s)$ we set 
    \[ 
    \gamma_i \coloneqq \left\{ 
            \begin{array}{ll} \gamma^R_i & \text{, for } i=1,\ldots,r \\ i \gamma^R_i  & \text{, for } i=r+1,\ldots,r+s
            \end{array} \right. \quad.
    \] 
    Then, the $\gamma_i$ obviously fulfill the Clifford relation in an orthonormal base of signature $(r,s)$. We set
    \begin{equation} < \label{eq:def_metric}
        \psi,\phi> \coloneqq  i^c \; \psi^\dagger   \cdot \gamma_{r+1} \ldots \gamma_{r+s} \cdot  \phi, 
    \end{equation}
    then an easy calculation shows that $c \in \mathds{N}$ can always be chosen such that Eq.~(\ref{eq:cond_metric}) holds for all ${ v=e_i, i=1,\ldots \dim(V)}$ and hence for all ${v \in V}$ (Here, it is important that for all $\gamma_i$ the same sign occurs.)
    
    \begin{remark}
    $h \coloneqq \gamma_{r+1} \ldots \gamma_{r+s}$ is the spin metric in the sense of  Refs.~\onlinecite{giesFermionsGravityLocal2014,lippoldtSpinbaseInvarianceFermions2015} for this special case. The need to have a spin metric is closely connected to the occurrence of anti-hermitian $\gamma$ matrices or, in a Lie-group view, to the non-compactness of $\so(r,s)$ for both $r,s \neq 0$, which prevents the existence of finite dimensional unitary representations. 
    \end{remark}
    
    \textit{\textbf{Uniqueness up to scale:}} Since the metric Eq.~(\ref{eq:def_metric}) is non-degenerate, any other metric must be of the form $<\psi,\phi>_{alt} = <\psi, m \cdot \phi>$ for some $m \in \glc$. From Eq.~(\ref{eq:cond_metric}) we conclude: $\gamma(v)  m = m \gamma(v) $ for all $v\in V$, hence $g m = m g $ for all $ g \in \glc$, so $m$ must be in the center of $\glc$,  i.e. a multiple of $\openone$.

 \section{Dirac structure and representations of Clifford bundle} \label{sec:dirac_cliff_bundle}
    For any pseudo-Riemannian manifold $(M,g)$ with signature $(r,s)$ there is a Clifford bundle $\cl(M)$ which is intuitively defined by taking this construction for each tangent space ${(T_qM,<,>_q)}$ for  all $q\in M$, where $<,>_q\coloneqq g_q(\cdot,\cdot)$. It can be checked that this indeed forms a smooth bundle over $M$ with typical fibre $\cl(r,s)$ which, as a vector bundle, but not as a bundle of algebras, is isomorphic to $\Lambda(M)$, the bundle of exterior forms over $M$. In particular, the Clifford bundle exists over any pseudo-Riemannian manifold, there is no specific obstruction (though, for clarity,  there are obstructions for an arbitrary manifold to carry a Lorentzian metric as opposed to a Riemannian metric, which always exists on a paracompact smooth  manifold.)

	We show that the existence of a Dirac structure is equivalent to the existence of a (fibre-wise) irreducible representation of the Clifford bundle on a complex vector bundle $E$:
	
	So first we assume the existence of such a representation:

	Since the Clifford bundle fibre over a point $x \in M$ is $\cl(T_xM)$, and $T_xM $ may be canonically identified with a subset  $ T_xM \subset \cl(T_xM)$ the representation of the Clifford bundle defines  a linear map from $T_xM$ to the endomorphism of $E$. Being a Clifford representation, it fulfills Eq.~(\ref{eq:dirac_str}). 
	
	Conversely, given a Dirac structure,  this defines a map from $T_xM$ to the endomorphisms of $E_x$. Again, using  $ T_xM \subset \cl(T_xM)$, the fact that $\cl(T_xM)$is generated by the product of such elements and the property Eq. (\ref{eq:dirac_str}), this map may be uniquely extended to a Clifford map from $ \cl(T_xM)$  to $\mathrm{End}(E_x)$, i.e. a representation of the Clifford algebra  on $E$ (in more categorial mathematical terms this follows from the universality property of Clifford algebras).

\section{Connections, covariant derivatives and principal bundles}	\label{sec:conn_principal}
    In this section, we give a short summary of some elementary constructions with principal bundles for easy reference for readers with little background in principal bundles and their associated vector bundles bundles. The exposition is informal without proofs, giving only ideas. Details may be found in   Refs.~\onlinecite{kobayashiFoundationsDifferentialGeometry2009, steenrodTopologyFibreBundles1999, nakaharaGeometryTopologyPhysics2018, eguchiGravitationGaugeTheories1980}.
    
    A principal bundle $P$ is a smooth fibre bundle over a base manifold $M$  whose fibre is a Lie group $G$ which acts by a free and transitive group action from the right. (Smoothness can be relaxed, but will be assumed throughout this article). Intuitively, this means we are attaching a copy of the group $G$ to each point in $M$ while ``forgetting where the identity lies'', i.e. there is no global section is this bundle. (It may be easily seen from the definition that a global section exists iff $P = M \times G$, the trivial bundle). 
    
    Given a vector bundle $E\rightarrow M$, a typical example for a principal bundle is the bundle of frames $F(E)$, where the fibre over $x$ consists of the set of all frames, i.e. of bases of the fibre $E_x$. This is a principal bundle with group $\mathrm{GL}(d,\mathds{R} )$ or $ \mathrm{GL}(d,\mathds{C} )$, depending on whether the vector bundle is real or complex.  If there is an additional structure like a metric on $E$, the set of frames may be restricted to orthogonal frames, yielding a principal bundle of orthogonal/unitary  frames with structure group $\mathrm{O}(d,\mathds{R} )$ or $ \mathrm{U}(d,\mathds{C} )$. If there is an orientation defined on the fibres of a real vector bundle  $E$, the structure group may be further reduced to $\mathrm{SO}(d,\mathds{R} )$  by only admitting positively oriented frames in the real. (Sligtly different construction for $SU(n)$ in the complex case.) 
    
    From the frame bundle, one may recover the original vector bundle by constructing the adjoint vector bundle in the real case as 
    \begin{equation} 
        \label{eq:assoc_vb} E \coloneqq ( P \times \real^d ) / \sim,
    \end{equation}
    where $(p g^{-1},g v) \sim (p,v)$ for all $g \in \mathrm{GL}(d,\mathds{R} )$. Intuitively, this means we can define a vector in $E$ by choosing a base and the components of the vector in this base, but we have to identify those pairs of base and coordinates, where the coordinates are transformed according to the base change.
    
    For an arbitrary principal bundle $P$ with structure group $G$, a vector space $V$ and a representation $\rho$ of $G$ on $V$, this construction generalizes in an obvious way to:
    \begin{equation} 
        \label{eq:gen_assoc_vb} E \coloneqq ( P \times V) / \sim, 
    \end{equation}
     with    $(p g^{-1},\rho(g) v) \sim (p,v)$  for all  $g \in G$.
    
    As is well known, a covariant derivative $\nabla$ on a vector bundle is a map 
    \begin{equation} \label{eq:cov_der}
        \nabla:     \Gamma(T M) \times \Gamma(E)    \rightarrow \Gamma(E), ~ (X,\Psi) \mapsto \nabla_X \Psi ,\end{equation}
    such that 
    \[ \begin{split}
    \nabla_X ( \Psi + f \Phi) &= \nabla_X \Psi  +  f \; \nabla_X \Phi + (X \cdot f) \;\Phi \\ \nabla_{(X+Y)} \Psi &= \nabla_X \Psi  + \nabla_Y \Psi,
    \end{split}\]
    for any vector fields $X,Y \in \Gamma(T M)$, sections $\Psi, \Phi$  of the vector bundle and function $f \in \mathcal{F}(M)$.
    
    Less used in most of physics literature, but well known to mathematicians, is the fact that this is essentially equivalent to two geometric constructions on principal  bundles, in particular the frame bundles:
    \begin{itemize}
        \item An invariant horizontal distribution over $P$, i.e. a smooth assignment of horizontal spaces $H_p \subset T_pP$ to all $p\in P$ such that 
        \[T_pP = H_p \oplus V_p ,  \qquad (R_g)_*(H_p) = H_{g p}.\]
        Here, $V_p$ is the vertical space, i.e. the tangent space to the fibres, and $(R_g)_*$ means the push-forward of a vector by the right multiplication $R_g$  by an element $g\in G$. Note that $V_p$ is spanned by the fundamental vector fields $\xi_P(p)$ corresponding to infinitesimal transformations  on $P$ of the Lie algebra element $\xi \in \mathfrak{G}$ of the Lie algebra corresponding to $G$.
        \item A connection form $\omega$ on $P$, i.e. a globally(!) defined one-form on $P$ with values in the Lie algebra $\mathfrak{G}$ such that
        \[ \omega(\xi_P(p)) =\xi \text{ for }  \xi \text{ in } \mathfrak{G},\qquad  R_g^*(\omega) = \mathrm{Ad}_{g^{-1}}\omega,\]
        where $\mathrm{Ad}$ is the adjoint representation of   $G$ on $\mathfrak{G}$, $R_g^*$ denotes  the pull-back of a one-form by the right multiplication $R_g$  by an element $g\in G$, and $\xi_P$ denotes the fundamental vector field generated by $\xi \in \mathfrak{G}$ as above.
    \end{itemize}
    
    The connection form $\omega$ is closely related, but not identical to the local connection form $\Gamma$ on $M$ found in most physics literature: The latter is the pullback to $M$ via a local section of $P$, hence its coordinate expression in a local trivialization. This explains why $\omega$ is a global object transforming tensorially, whereas $\Gamma$ is only locally defined  and has the ``weird'' transformation behaviour of connection forms under change of local trivializing sections.
    
    Both constructions are equivalent and uniquely define a covariant derivative on an associated vector bundle. A proof may be found in  Ref.~\onlinecite{kobayashiFoundationsDifferentialGeometry2009}. Essentially, the logic is as follows: Given a horizontal distribution, we may uniquely decompose any vector  $X\in T_pP$ into its horizontal and vertical part. Now, $\omega(X_p)$ is simply defined as the element $\xi \in  \mathfrak{G}$ such that its fundamental vector field $\xi_P(p)$ is just the vertical part of $X_p$. It can be checked that $\omega$ defined in this way is indeed a connection form.
    
    Conversely, given $\omega$, we simply define $H_p \coloneqq \ker(\omega(p))$, i.e. the horizontal vectors are precisely those $X_p \in T_pP$ which are mapped to $0$ by $\omega$.
    
    Finally, a covariant derivative on an associated $E$ may be defined from one of those two constructions, e.g. in a ``physicist's fashion'' by choosing a local section in $P$, pulling back $\omega$ to $M$ and using the standard formulas for a covariant derivative in terms of an ordinary derivative and a local connection form and showing that this fulfills the definition of a covariant derivative and is independent of the choice of section. Alternatively, the following more geometric approach may be used:
    
    Let $X$ be a vector field and $\Psi$ be a section in the vector bundle $E$, $\chi(t)$ an integral curve of $X$ in $M$ tangent to $X$, i.e. a curve with tangent $\frac{d}{dt} \chi(t) = X(\chi(t))$. Then, using the horizontal distribution, we may define for any $p\in P$ which projects onto the curve origin $\chi(0)$ a unique horizontal lift $\tilde{\chi}(t)$  for $t\in [0,\epsilon]$ for some $\epsilon > 0$  such that the tangent vector is horizontal (i.e. in $H_{\tilde{\chi}(t)}$) for all $t\in [0,\epsilon]$. 
    
    Now, $\tilde{\chi}$ picks a specific representative of the equivalence class of  $\sim$ in Eq.~(\ref{eq:gen_assoc_vb}), hence $\Psi|_\chi$ may be identified with a function $\psi: t \mapsto V$. Now, to define the covariant derivative at the point $\chi(0)$, we simply take the ordinary derivative in $V$ and the equivalence class module $\sim$ of  $(\tilde{\chi}(0), \frac{d}{dt} \psi(0) )$.
    
    (Intuitively, what we have done is choosing a base which is covariantly constant along the curve $\chi$, decomposing $\Psi$ in this base where the coefficients now are just functions, and using Eq.~(\ref{eq:cov_der})  to express $\nabla_X \Psi$ by an ordinary derivative in this adjusted base.)
    
    Hence, for an arbitrary principal bundle, a connection in one of the two constructions uniquely defines a covariant derivative on any associated vector bundle. In case the principal bundle is a frame bundle of $E$, the converse holds as well: We may construct a horizontal distribution from $\nabla$ defining a connection on the principal bundle. This can be easily seen by considering for $p\in P$ a curve $\chi$ in $M$ with starting point $x = \pi(p) \in M$. Then, we may consider the horizontal lift $\tilde{\gamma}$, a curve in $P$ defined by parallel transporting all basis vectors of the frame  according to $\nabla$    along $\gamma$, thus defining a frame over each point of $\gamma$, and its tangent vector at $p$ (corresponding to $t=0$) is by definition horizontal. By considering a whole set of curves $\chi$ such that its tangent span $T_xM$, we may lift the whole $T_xM$ to a horizontal subspace $H_p \subset T_pP$. This fulfills the requirement of a connection, hence covariant derivative and connection on the frame bundle are equivalent. (This also works if the structure group of the frame bundle is restricted to some subgroup of $\mathrm{GL(d,\real})$.
    
    We will use this general  equivalence in section \ref{sec:metric_conn} to show the existence of a compatible connection in a very straightforward way.

\section{Elementary proof of triviality of $\real^+$ bundles} \label{sec_elem_triv}
    We consider an $\real^+$ principal bundle $\pi: L \rightarrow M$ over a paracompact manifold $M$, where $\real^+$ is a group under multiplication. Since $\real^+$ is contractible, it is a well known result from algebraic topology that such a bundle is always trivial (see e.g.   Ref.~\onlinecite{hatcherAlgebraicTopology2002}). However, for easy reference we give here an elementary proof using a partition of unity argument without using arguments from algebraic topology:
    
    First we note that by the group isomorphism $\real \rightarrow \real^+, t \mapsto e^t $ of the additive group $\real$ and the multiplicative group $\real$, we may equivalently consider $L$  a principal bundle over the additive group $\real$. (Despite the additive structure, this is not a priori a one dimensional vector bundle, but a principal bundle, since the transition functions act via translations, i.e. there is no preferred zero-section).
    
    Let $(V_i)_{i\in I}$ be a locally finite covering of $M$ and $\chi_i$ be a partition of unity, i.e. a set of non-negative maps $\chi_i: M \rightarrow [0,1]$ such that the support of $\chi_i$ is contained in $V_i$ and 
    \begin{equation} \label{eq:part_un}
        \sum _i \chi_i=1 ~.
    \end{equation}
    By refinement of $V_i$ (if necessary) we may assume that $L$ is locally trivial over each $V_i$. Hence, we may choose for each $i$ a local section $s_i$ in $L$. 
    
    Despite the fact that $L$ is  not a vector bundle, we may define a global section $s \coloneqq \sum \chi_i s^i $ due to the fact that $\sum _i \chi_i=1$:
    This is defined on a neighbourhood $U = V_i$ for a fixed $i \in I$ as follows:
    \begin{equation} \label{eq:def_sect}
    s^U(x) \coloneqq \sum \chi_i(x) (s^i)^U(x) , 
    \end{equation}
    where $s^U , (s^i)^U: U \rightarrow \real$ denote the real valued functions corresponding to the sections $s, s^i$ under the chosen local trivialization  over $U$. 
    This indeed defines a global section $s \in \Gamma(L)$ since, if $W = V_k$ is another neighbourhood with non-empty intersection with $U$ and $\alpha_{UW}: U\cap W \rightarrow \real$ the corresponding transition function, then:
    \begin{align*}
        s^W(x) & \coloneqq \sum \chi_i(x) (s^i)^W(x) = \sum \chi_i(x) \left( (s^i)^U(x) + \alpha_{UW}(x)\right) \\ 
        & = \sum \chi_i(x)  (s^i)^U(x) + \sum \chi_i(x)  \alpha_{UW}(x) 
        \\ &\overset{\mathrm{(\ref{eq:part_un}), (\ref{eq:def_sect}})}{=}
    s^U(x) + \alpha_{UW}(x) ,
    \end{align*}
    so $s$ transforms as a well-defined section under the transition functions.

\bibliography{bibliography}

\end{document}